\documentclass[11pt, a4paper]{article}

\usepackage[utf8]{inputenc}
\usepackage[T1]{fontenc}
\usepackage[english]{babel}
\usepackage{amsmath, amssymb, amsthm, mathtools}
\usepackage{enumitem}
\usepackage{geometry}
\usepackage[colorlinks=true, allcolors=blue]{hyperref}
\usepackage[capitalize]{cleveref}
\usepackage[numbers, sort&compress]{natbib}

\theoremstyle{plain}
\newtheorem{theorem}{Theorem}
\newtheorem{lemma}[theorem]{Lemma}
\newtheorem{proposition}[theorem]{Proposition}
\newtheorem{corollary}[theorem]{Corollary}
\theoremstyle{definition}

\theoremstyle{remark}
\newtheorem*{remark}{Remark}

\renewcommand{\d}{\operatorname{d}\!}
\newcommand{\E}{\mathbb E}
\newcommand{\Pp}{\mathbb P}
\newcommand{\R}{\mathbb R}
\newcommand{\Z}{\mathbb Z}
\newcommand{\Lam}{\Lambda}
\newcommand{\cP}{\mathcal P}
\newcommand{\fr}{\mathrm{fr}}
\newcommand{\per}{\mathrm{per}}
\newcommand{\ebf}{\mathbf e}

\title{\Large\bf Wasserstein Stability, Couplings Across Volumes, and the $1+1/d$ Moment Thresholds in the Edwards--Anderson Model}
\author{Mauris Chueng\thanks{These authors contributed equally to this work.}\ \textsuperscript{,}\thanks{Corresponding author: maurischueng@gmail.com} \and Hexiang Wang\footnotemark[1] \and Keheng Zhu\footnotemark[1]}
\date{}

\newcommand{\Addresses}{{%
    \bigskip
    \footnotesize

    \textsc{Mauris Chueng}, \textsc{School of Statistics and Data Science, Jilin University of Finance and Economics, Changchun, 130117, China}\par\nopagebreak
    \texttt{maurischueng@gmail.com}
    \medskip

    \textsc{Hexiang Wang}, \textsc{School of Mathematical Sciences, Nankai University, Tianjin, 300071, China}\par\nopagebreak
    \texttt{Kui6539@outlook.com}
    \medskip

    \textsc{Keheng Zhu}, \textsc{Academy for Multidisciplinary Studies and School of Mathematical Sciences, Capital Normal University, Beijing, 100048, China}\par\nopagebreak
    \texttt{hexistartop@gmail.com}
    \medskip
}}

\begin{document}

\maketitle
\begin{abstract}
We study the quenched pressure of the nearest-neighbor Edwards--Anderson Ising model with free and periodic boundary conditions. First, we prove that the infinite-volume pressure is $\beta d$-Lipschitz in the coupling law for the $1$-Wasserstein distance, yielding quantitative thermodynamic limits for spatially inhomogeneous disorder. Second, for periodic volumes, we prove that almost-sure convergence under every joint coupling of the finite-volume disorder arrays is equivalent to complete convergence of the one-volume pressure laws. Third, we prove that $\mathbb{E}\vert{}J\vert{}^{1+1/d}<\infty$ guarantees this universal-coupling conclusion. We further show that this exponent is optimal among uniform power-moment assumptions: for every $1\leq q<1+1/d$, there is a centered symmetric law with finite $q$-th moment for which canonical nested volumes converge almost surely, whereas independently resampled volumes with the same fixed-volume marginals converge in probability but not almost surely. Finally, in dimension one, we prove that the first-moment condition is also necessary for a finite limiting pressure.
\end{abstract}



\tableofcontents

\section{Introduction}

The Edwards--Anderson model \cite{EA1975} is the standard finite-dimensional random-bond Ising model. Existence of its quenched thermodynamic limit belongs to a well-developed theory of disordered lattice systems; representative results include \cite{Zegarlinski1991,CGP2004,ContucciGraffi2004,CS2009}. The influence of free, periodic, and antiperiodic boundary conditions on the surface pressure was studied in \cite{CGSurface2004}, and periodic boundary stability is part of a broader principle going back to \cite{FisherLebowitz1970}. We therefore use the i.i.d. thermodynamic limit only as a baseline and focus on two questions not answered by its usual formulation.

Firstly, how stable is the limiting pressure when the disorder distribution changes with the edge? We prove a finite-volume transport estimate and show that the limiting pressure is Lipschitz on the space of coupling laws with finite first moment, equipped with the $1$-Wasserstein distance. This gives an explicit homogenization statement for independent but non-identically distributed couplings: if their local laws approach a reference law in spatially averaged transport distance, the quenched pressures approach the reference thermodynamic limit. A common quantile construction and a Kolmogorov condition give a pathwise version.

Secondly, what does ``almost surely as $L\to\infty$'' mean when a new disorder array is drawn at every volume? The one-volume marginals determine convergence in probability but not almost sure convergence. Using the classical notion of complete convergence introduced by Hsu and Robbins \cite{HsuRobbins1947}, we give a necessary and sufficient condition for convergence under every joint coupling of the volume arrays. A block decomposition and the critical Baum--Katz theorem \cite{BaumKatz1965} then show that the finite moment $\E|J|^{1+1/d}<\infty$ is sufficient. Below this exponent we construct finite-$q$-moment laws for which the canonical nested construction converges almost surely, while independent arrays across volumes have infinitely many order-one pressure spikes. Thus the exponent $1+1/d$ is sharp on the moment scale.

The transport estimate itself is elementary; the contribution is the combination of its thermodynamic, spatially inhomogeneous, and coupling-theoretic consequences with the sharp moment-scale threshold and a family of counterexamples below it. These statements do not appear to have been recorded together in this form. No claim of novelty is made for the i.i.d. existence theorem in \cref{sec:i.i.d.} or for the classical Baum--Katz theorem used at the endpoint.

The logical division is as follows. After deterministic comparisons and a
self-contained i.i.d. baseline, \cref{thm:wasserstein} establishes transport
stability of both finite- and infinite-volume pressures.
\Cref{thm:inhomogeneous} turns this estimate into a quantitative
thermodynamic limit for spatially varying laws. The universal-coupling
criterion in \cref{thm:universal-coupling} is then combined with a
Baum--Katz block argument in \cref{thm:moment-criterion}. Sharpness below
the endpoint is established by \cref{thm:heavy-tail}. This counterexample is not a
reformulation of the i.i.d. theorem, because it compares different joint laws
across volumes having exactly the same one-volume marginals.

\section{Model and deterministic comparisons}\label{sec:model}

Fix $d\geq1$ and let
\[
 \Lam_L=\{0,\ldots,L-1\}^d,\qquad |\Lam_L|=L^d.
\]
For $x\in\Z_{\geq0}^d$ and $1\leq i\leq d$, the pair $(x,i)$ labels the positively oriented edge from $x$ to $x+\ebf_i$. Let $\nu$ be a probability law on $\R$ satisfying
\begin{equation}\label{eq:first-moment}
	m_1(\nu):=\int_{\R}|t|\,\nu(dt)<\infty,
\end{equation}
and let $\{J_{x,i}\}$ be i.i.d. with law $\nu$ on a common probability space. For $L\geq1$, the free Hamiltonian is
\[
 H_L^{\fr}(\sigma;J)
 =-\sum_{i=1}^d\sum_{\substack{x\in\Lam_L\\x_i\leq L-2}}
 J_{x,i}\sigma_x\sigma_{x+\ebf_i}.
\]
For $L\geq3$, the periodic Hamiltonian is
\[
 H_L^{\per}(\sigma;J)
 =-\sum_{i=1}^d\sum_{x\in\Lam_L}
 J_{x,i}\sigma_x\sigma_{x+\ebf_i\bmod L}.
\]
The restriction $L\geq3$ keeps the periodic interaction graph simple. For $b\in\{\fr,\per\}$ and $\beta\geq0$, set
\[
 Z_L^b(\beta,J)=\sum_{\sigma\in\{-1,1\}^{\Lam_L}}
 e^{-\beta H_L^b(\sigma;J)},\qquad
 \cP_L^b(\beta,J)=\frac{1}{L^d}\log Z_L^b(\beta,J),
\]
and write $p_{L,\nu}^b(\beta)=\E\cP_L^b(\beta,J)$.

We first record two deterministic estimates.

\begin{lemma}\label{lem:comparison}
Let $\Omega$ and $I$ be finite, let $S_e:\Omega\to[-1,1]$, and define
\[
 Z(K)=\sum_{\sigma\in\Omega}
 \exp\!\left(\beta\sum_{e\in I}K_eS_e(\sigma)\right).
\]
For coefficient families $K$ and $K'$,
\begin{equation}\label{eq:coefficient-comparison}
	|\log Z(K)-\log Z(K')|
	\leq\beta\sum_{e\in I}|K_e-K'_e|.
\end{equation}
\end{lemma}

\begin{proof}
The two exponents differ pointwise by at most
$a=\beta\sum_e|K_e-K'_e|$. Hence
$e^{-a}Z(K')\leq Z(K)\leq e^aZ(K')$.
\end{proof}

In particular, on a finite graph $G=(V,E)$,
\begin{equation}\label{eq:basic-bound}
	|V|\log2-\beta\sum_{e\in E}|J_e|
	\leq\log Z_G(\beta,J)
	\leq |V|\log2+\beta\sum_{e\in E}|J_e|.
\end{equation}
Thus all pressures considered below are integrable whenever the relevant edge laws have finite first moments.

\begin{lemma}\label{lem:max-bond}
For every finite simple graph $G=(V,E)$,
\begin{equation}\label{eq:max-bond}
	\log Z_G(\beta,J)\geq\beta\max_{e\in E}|J_e|,
\end{equation}
where the maximum over the empty set is zero.
\end{lemma}

\begin{proof}
Choose $e_0=\{u,v\}$ maximizing $|J_e|$ and let
$s=\operatorname{sign}(J_{e_0})$; the case $J_{e_0}=0$ is immediate. Under the uniform measure on the nonempty finite set of configurations satisfying $\sigma_u\sigma_v=s$, every other edge product has mean zero. Indeed, an edge different from $e_0$ either is disjoint from $e_0$ or shares exactly one endpoint with it, and in both cases at least one unconstrained symmetric spin remains. The conditional average of
$\sum_{e\ne e_0}J_e\sigma_e$ is therefore zero, so some constrained configuration makes this sum nonnegative. Its full interaction energy is at least $|J_{e_0}|$, and the corresponding Boltzmann weight proves \eqref{eq:max-bond}.
\end{proof}

\section{The i.i.d. baseline}\label{sec:i.i.d.}

For completeness, we prove the classical i.i.d. conclusion in the precise form
used later. The proof is retained because its explicit constants feed into
the transport and complete-convergence estimates; no originality is claimed
for existence of the i.i.d. thermodynamic limit.

Fix integers $L\geq k\geq1$, write $r=\lfloor L/k\rfloor$, and set
\[
 C_{L,k}=\{0,\ldots,rk-1\}^d,
 \qquad R_{L,k}=\Lam_L\setminus C_{L,k}.
\]
Partition $C_{L,k}$ into the $r^d$ translates
$B_z=kz+\Lam_k$, where $z\in\{0,\ldots,r-1\}^d$. Delete every free edge
that is not internal to one of the blocks, and denote the deleted set by
$D_{L,k}$.

\begin{lemma}\label{lem:deleted-count}
For $L\geq k\geq1$,
\begin{equation}\label{eq:deleted-count}
	\frac{|D_{L,k}|}{L^d}
	\leq\frac{d}{k}+\frac{2d^2k}{L}.
\end{equation}
\end{lemma}

\begin{proof}
The interfaces between adjacent blocks inside $C_{L,k}$ contribute
$d(r-1)(rk)^{d-1}$ edges. Every other deleted edge has at least one endpoint
in $R_{L,k}$, so there are at most $2d|R_{L,k}|$ of them. Since
\[
 |R_{L,k}|=L^d-(rk)^d
 \leq d(L-rk)L^{d-1}\leq dkL^{d-1},
\]
division by $L^d$ proves \eqref{eq:deleted-count}.
\end{proof}

The decoupled partition function is
\begin{equation}\label{eq:decoupled-partition}
	Z_{L,k}^{\mathrm{dec}}
	=2^{|R_{L,k}|}
	\prod_{z\in\{0,\ldots,r-1\}^d}Z_{B_z}^{\fr}.
\end{equation}

\begin{proposition}\label{prop:mean-free}
For every $\beta\geq0$, the finite limit
\[
 p_{\infty,\nu}(\beta)=\lim_{L\to\infty}p_{L,\nu}^{\fr}(\beta)
\]
exists, and
\begin{equation}\label{eq:free-rate}
	|p_{k,\nu}^{\fr}(\beta)-p_{\infty,\nu}(\beta)|
	\leq\frac{\beta d m_1(\nu)}{k}.
\end{equation}
\end{proposition}

\begin{proof}
Apply \cref{lem:comparison} to the full and decoupled systems, take
expectations, and divide by $L^d$. Translation invariance and
\eqref{eq:decoupled-partition} give
\begin{align}
	\left|p_{L,\nu}^{\fr}
	-\left(\frac{rk}{L}\right)^dp_{k,\nu}^{\fr}
	-\frac{|R_{L,k}|}{L^d}\log2\right|
	&\leq\beta m_1(\nu)\frac{|D_{L,k}|}{L^d}\notag\\
	&\leq\beta m_1(\nu)
	\left(\frac{d}{k}+\frac{2d^2k}{L}\right).
	\label{eq:mean-tiling}
\end{align}
For fixed $k$, send $L\to\infty$. Since $rk/L\to1$ and
$|R_{L,k}|/L^d\to0$,
\[
 p_{k,\nu}^{\fr}-\frac{\beta d m_1(\nu)}{k}
 \leq\liminf_Lp_{L,\nu}^{\fr}
 \leq\limsup_Lp_{L,\nu}^{\fr}
 \leq p_{k,\nu}^{\fr}+\frac{\beta d m_1(\nu)}{k}.
\]
The sequence is bounded by \eqref{eq:basic-bound}. Letting $k\to\infty$
shows that its limsup and liminf agree. Taking the large-volume limit in the
same display proves \eqref{eq:free-rate}.
\end{proof}
\noindent We next isolate the strong-law statements needed for the sample limit. Let
$Y_{x,i}=|J_{x,i}|$.

\begin{lemma}\label{lem:spatial-slln}
There is an event of probability one on which all the following hold:
\begin{enumerate}[label=\textup{(\roman*)}]
\item
\[
 \frac1{L^d}\sum_{i=1}^d\sum_{x\in\Lam_L}Y_{x,i}
 \longrightarrow dm_1(\nu);
\]
\item for every fixed $k\geq1$, direction $i$, and residue
$a\in\{0,\ldots,k-1\}$,
\[
 \frac1{L^d}
 \sum_{\substack{x\in\Lam_L\\x_i\equiv a\pmod k}}Y_{x,i}
 \longrightarrow\frac{m_1(\nu)}{k};
\]
\item for every fixed $h\geq1$, with $\Lam_{L-h}=\varnothing$ when $L\leq h$,
\[
 \frac1{L^d}\sum_{i=1}^d
 \sum_{x\in\Lam_L\setminus\Lam_{L-h}}Y_{x,i}
 \longrightarrow0.
\]
\end{enumerate}
\end{lemma}

\begin{proof}
For each direction, enumerate $\Z_{\geq0}^d$ shell by shell so that the first
$L^d$ sites are $\Lam_L$. The ordinary strong law gives (i). Enumerating a
fixed residue class in the same order proves (ii), because the number of its
sites in $\Lam_L$, divided by $L^d$, tends to $1/k$. For (iii), let
$A_L=\sum_{i,x\in\Lam_L}Y_{x,i}$. Part (i) gives
\[
 \frac{A_L-A_{L-h}}{L^d}
 =\frac{A_L}{L^d}
 -\left(\frac{L-h}{L}\right)^d
 \frac{A_{L-h}}{(L-h)^d}
 \longrightarrow0.
\]
There are only countably many choices of $k,i,a,h$, so the corresponding
probability-one events may be intersected.
\end{proof}

\begin{proposition}\label{prop:sample-free}
For every fixed $\beta\geq0$,
\[
 \cP_L^{\fr}(\beta,J)\stackrel{a.s.}{\longrightarrow} p_{\infty,\nu}(\beta)
\]
\end{proposition}

\begin{proof}
Fix $k$. The block variables
$X_z^{(k)}=\log Z_{B_z}^{\fr}$ are i.i.d. and integrable, because distinct blocks
use disjoint internal edges. The strong law, applied in shell order, gives
\begin{equation}\label{eq:block-slln}
	\frac1{r^d}\sum_{z\in\{0,\ldots,r-1\}^d}X_z^{(k)}
	\stackrel{a.s.}{\longrightarrow} k^dp_{k,\nu}^{\fr}.
\end{equation}
Since $r^d/L^d\to k^{-d}$, the same block sum divided by $L^d$ converges to
$p_{k,\nu}^{\fr}$. The deleted edges inside $C_{L,k}$ have origins in the
residue class $x_i\equiv k-1\pmod k$ in each direction; every remaining
deleted edge has an origin in a shell of width $k$. Hence
\cref{lem:spatial-slln} gives
\begin{equation}\label{eq:deleted-slln}
	\limsup_{L\to\infty}\frac1{L^d}
	\sum_{e\in D_{L,k}}|J_e|
	\leq\frac{dm_1(\nu)}{k}.
\end{equation}
The pointwise coefficient comparison, \eqref{eq:decoupled-partition},
\eqref{eq:block-slln}, and \eqref{eq:deleted-slln} now imply
\[
	\limsup_{L\to\infty}
	|\cP_L^{\fr}-p_{k,\nu}^{\fr}|
	\leq\frac{\beta dm_1(\nu)}{k}.
\]
Intersect over $k\in\mathbb N$, use \eqref{eq:free-rate}, and let
$k\to\infty$.
\end{proof}

\begin{theorem}\label{thm:i.i.d.}
If $m_1(\nu)<\infty$, then for every $\beta\geq0$ there exists a finite
deterministic limit $p_{\infty,\nu}(\beta)$ such that
\[
	p_{L,\nu}^{\fr}(\beta),\ p_{L,\nu}^{\per}(\beta)
	\longrightarrow p_{\infty,\nu}(\beta).
\]
For $L\geq3$,
\begin{equation}\label{eq:i.i.d.-rates}
	|p_{L,\nu}^{\fr}-p_{\infty,\nu}|
	\leq\frac{\beta d m_1(\nu)}{L},\qquad
	|p_{L,\nu}^{\per}-p_{\infty,\nu}|
	\leq\frac{2\beta d m_1(\nu)}{L}.
\end{equation}
Equivalently, for $\beta>0$ and $f=-p/\beta$, the corresponding free-energy
density errors are bounded by $dm_1(\nu)/L$ and $2dm_1(\nu)/L$.
On the canonical disorder space,
\[
	\cP_L^{\fr}(\beta,J),\ \cP_L^{\per}(\beta,J)
	\stackrel{a.s.}{\longrightarrow} p_{\infty,\nu}(\beta).
\]
There is one probability-one event on which both convergences hold locally
uniformly for $\beta\in[0,\infty)$.
\end{theorem}

\begin{proof}
The free mean and sample assertions are
\cref{prop:mean-free,prop:sample-free}. The periodic system differs from the
free one by the wrap-around set
\[
	S_L=\{(x,i):x\in\Lam_L,\ x_i=L-1,\ 1\leq i\leq d\},
	\qquad |S_L|=dL^{d-1}.
\]
Consequently,
\begin{equation}\label{eq:seam-bound}
	|\cP_L^{\per}-\cP_L^{\fr}|
	\leq\frac{\beta}{L^d}\sum_{(x,i)\in S_L}|J_{x,i}|.
\end{equation}
The expectation is at most $\beta dm_1(\nu)/L$, which combines with
\eqref{eq:free-rate} to prove \eqref{eq:i.i.d.-rates}. Every origin in $S_L$
lies in $\Lam_L\setminus\Lam_{L-1}$, so the right side of
\eqref{eq:seam-bound} tends to zero almost surely by
\cref{lem:spatial-slln}. This proves the periodic sample limit.

For either boundary condition and $\beta,\beta'\geq0$,
\[
	|\cP_L^b(\beta,J)-\cP_L^b(\beta',J)|
	\leq|\beta-\beta'|\frac1{L^d}
	\sum_{i=1}^d\sum_{x\in\Lam_L}|J_{x,i}|.
\]
The random Lipschitz constants converge almost surely to $dm_1(\nu)$ by
\cref{lem:spatial-slln}; the deterministic pressures and their limit are
$dm_1(\nu)$-Lipschitz. Intersect the sample-limit events over nonnegative
rational $\beta$ and use a finite rational net on each compact interval.
\end{proof}

\begin{remark}\label{rem:no-symmetry}
No centering, symmetry, variance, or exponential moment is used in
\cref{thm:i.i.d.}. The first absolute moment makes the block pressures
integrable and the deleted and boundary-edge sums negligible per site.
\end{remark}

\section{Wasserstein Stability}\label{sec:transport}

Let $\mathcal P_1(\R)$ be the probability laws with finite first moment. For $\nu,\mu\in\mathcal P_1(\R)$, define
\[
	W_1(\nu,\mu)=\inf_{\pi\in\Pi(\nu,\mu)}
	\int_{\R^2}|x-y|\,\pi(\d x,\d y).
\]
All random pressure laws below also belong to $\mathcal P_1(\R)$ by \eqref{eq:basic-bound}.

\begin{theorem}\label{thm:wasserstein}
For $b\in\{\fr,\per\}$ and every admissible $L$,
\begin{align}
	|p_{L,\nu}^b(\beta)-p_{L,\mu}^b(\beta)|
	&\leq\beta\frac{|E_L^b|}{L^d}W_1(\nu,\mu)
	\leq\beta dW_1(\nu,\mu),\label{eq:mean-w1}\\
	W_1\!\left(\operatorname{Law}(\cP_{L,\nu}^b),
	\operatorname{Law}(\cP_{L,\mu}^b)\right)
	&\leq\beta dW_1(\nu,\mu).\label{eq:law-w1}
\end{align}
Consequently,
\begin{equation}\label{eq:limit-w1}
	|p_{\infty,\nu}(\beta)-p_{\infty,\mu}(\beta)|
	\leq\beta dW_1(\nu,\mu).
\end{equation}
\end{theorem}

\begin{proof}
Take an optimal coupling $(X,Y)$ of $\nu$ and $\mu$; on the real line it may be realized by common quantiles \cite[Chapter~2]{Villani2009}. Use independent copies $(X_e,Y_e)$ over the finite edge set. By \cref{lem:comparison},
\[
 |\cP_L^b(X)-\cP_L^b(Y)|
 \leq\frac{\beta}{L^d}\sum_{e\in E_L^b}|X_e-Y_e|.
\]
Expectation gives \eqref{eq:mean-w1}. The same construction is an admissible coupling of the two random pressures, proving \eqref{eq:law-w1}. Passing to the limits supplied by \cref{thm:i.i.d.} gives \eqref{eq:limit-w1}.
\end{proof}

The underlying deterministic consequence will be used repeatedly: if two disorder fields on a common space satisfy
\[
 \frac{1}{L^d}\sum_{x\in\Lam_L}\sum_{i=1}^d
 |J_{x,i}-K_{x,i}|\longrightarrow0
\]
in probability or almost surely, then their free and periodic pressures have the same asymptotic limit in the corresponding mode of convergence.

\section{Asymptotically homogeneous disorder}\label{sec:inhomogeneous}

For each positive directed edge $e=(x,i)$, let $\nu_e\in\mathcal P_1(\R)$, and suppose the edge couplings are spatially independent with these laws. Fix a reference law $\nu\in\mathcal P_1(\R)$ and define
\begin{equation}\label{eq:delta-L}
	\delta_L(\nu_\bullet,\nu)
	=\frac{1}{L^d}\sum_{x\in\Lam_L}\sum_{i=1}^d
	W_1(\nu_{x,i},\nu).
\end{equation}
Let $p_{L,\nu_\bullet}^b=\E\cP_L^b$ denote the corresponding quenched pressure.

For a distribution function $F$, write
$F^{-1}(u)=\inf\{t\in\R:F(t)\geq u\}$. Let $\{U_e\}$ be i.i.d. uniform variables and define the common quantile coupling
\begin{equation}\label{eq:quantile-coupling}
	J_e=F_e^{-1}(U_e),\qquad K_e=F^{-1}(U_e).
\end{equation}
Then the $J_e$ are independent with laws $\nu_e$, the $K_e$ are i.i.d. with law $\nu$, and
$\E|J_e-K_e|=W_1(\nu_e,\nu)$.

\begin{theorem}\label{thm:inhomogeneous}
Suppose $\delta_L(\nu_\bullet,\nu)\to0$. Then
\begin{align}
	|p_{L,\nu_\bullet}^{\fr}(\beta)-p_{\infty,\nu}(\beta)|
	&\leq\beta\delta_L(\nu_\bullet,\nu)
	+\frac{\beta d m_1(\nu)}{L},\label{eq:inhom-free}\\
	|p_{L,\nu_\bullet}^{\per}(\beta)-p_{\infty,\nu}(\beta)|
	&\leq\beta\delta_L(\nu_\bullet,\nu)
	+\frac{2\beta d m_1(\nu)}{L}.
	\label{eq:inhom-periodic}
\end{align}
The random free and periodic pressures converge in probability to
$p_{\infty,\nu}(\beta)$.

Enumerate the directed edges shell by shell so that the first $dL^d$ edges are $\Lam_L\times\{1,\ldots,d\}$, and let $D_e=|J_e-K_e|$ in \eqref{eq:quantile-coupling}. If
\begin{equation}\label{eq:kolmogorov-condition}
	\sum_{n=1}^{\infty}\frac{\operatorname{Var}(D_{e_n})}{n^2}<\infty,
\end{equation}
then both pressures converge almost surely on the common quantile coupling. In particular, \eqref{eq:kolmogorov-condition} follows from
$\sup_e\E D_e^2<\infty$.
\end{theorem}

\begin{proof}
Apply \cref{lem:comparison} to \eqref{eq:quantile-coupling}. The periodic edge set is exactly $\Lam_L\times\{1,\ldots,d\}$, and the free edge set is contained in it, so
\begin{equation}\label{eq:inhom-coupling-bound}
	\E|\cP_L^b(J)-\cP_L^b(K)|
	\leq\beta\delta_L(\nu_\bullet,\nu).
\end{equation}
Combining this with \eqref{eq:i.i.d.-rates} proves
\eqref{eq:inhom-free}--\eqref{eq:inhom-periodic}. Markov's inequality makes the difference in \eqref{eq:inhom-coupling-bound} tend to zero in probability, while $\cP_L^b(K)$ converges almost surely by \cref{thm:i.i.d.}. This proves convergence in probability. Since it is a statement about the one-volume laws, it holds for every joint construction across volumes having those marginals.

Under \eqref{eq:kolmogorov-condition}, Kolmogorov's strong law for independent non-identically distributed variables \cite[Chapter~2]{Durrett2019} gives
\[
	\frac{1}{N}\sum_{n=1}^N(D_{e_n}-\E D_{e_n})\stackrel{a.s.}{\longrightarrow}.
\]
Taking $N=dL^d$, multiplying by $d$, and using \eqref{eq:delta-L} yields
\[
	\frac{1}{L^d}\sum_{x\in\Lam_L}\sum_{i=1}^dD_{x,i} \stackrel{a.s.}{\longrightarrow}0.
\]
The deterministic comparison and the canonical i.i.d. limit complete the proof.
\end{proof}

\begin{corollary}\label{cor:sparse}
Suppose there are deterministic sets
$A_L\subset\Lam_L\times\{1,\ldots,d\}$ such that
\[
 \{e\in\Lam_L\times\{1,\ldots,d\}:\nu_e\ne\nu\}\subseteq A_L,
 \qquad |A_L|=o(L^d),
\]
and
\[
	M:=\sup_{e\in\Z_{\geq0}^d\times\{1,\ldots,d\}} W_1(\nu_e,\nu)<\infty.
\]
Then the quenched-mean bounds and the convergence-in-probability conclusion
of \cref{thm:inhomogeneous} hold. No nesting of the sets $A_L$ is required.
The almost-sure conclusion is not asserted unless
\eqref{eq:kolmogorov-condition} is also satisfied.
\end{corollary}

\begin{proof}
Let
\[
	C_L=\{e\in\Lam_L\times\{1,\ldots,d\}:\nu_e\ne\nu\}.
\]
For $e\notin C_L$, one has $W_1(\nu_e,\nu)=0$. Since $C_L\subseteq A_L$,
\[
	\delta_L(\nu_\bullet,\nu)=\frac{1}{L^d}\sum_{e\in C_L}W_1(\nu_e,\nu)\leq M\frac{|A_L|}{L^d}\longrightarrow 0.
\]
The result follows from \cref{thm:inhomogeneous}.
\end{proof}

\section{Joint couplings across volumes}\label{sec:volume-couplings}

Fix $\nu\in\mathcal P_1(\R)$, $\beta\geq0$, and let $Q_L$ be the law of the i.i.d. periodic pressure $\cP_L^{\per}(\beta,J)$. By \cref{thm:i.i.d.},
\begin{equation}\label{eq:marginal-convergence}
	Q_L\bigl(\{x:|x-p_{\infty,\nu}(\beta)|>t\}\bigr)\longrightarrow0
	\qquad(t>0).
\end{equation}
We say these pressures converge completely if the probabilities in
\eqref{eq:marginal-convergence} are summable in $L$ for every $t>0$.

\begin{theorem}[Universal coupling criterion]\label{thm:universal-coupling}
The following are equivalent.
\begin{enumerate}[label=\textup{(\roman*)}]
\item For every joint coupling of the finite-volume i.i.d. disorder arrays,
$\cP_L^{\per}\to p_{\infty,\nu}(\beta)$ almost surely.
\item The periodic pressure laws converge completely:
\[
	\sum_{L=3}^{\infty}Q_L\bigl(\{x:|x-p_{\infty,\nu}(\beta)|>t\}\bigr)<\infty	\qquad(t>0).
\]
\end{enumerate}
\end{theorem}

\begin{proof}
If (ii) holds, the first Borel--Cantelli lemma applies under every joint coupling; no independence across volumes is required. Intersecting the resulting probability-one events over positive rational $t$ proves (i).

If (ii) fails for some $t$, take the disorder arrays independently for different $L$. The corresponding deviation events are independent and have a divergent probability sum. The second Borel--Cantelli lemma implies that deviations of size $t$ occur infinitely often almost surely, contradicting (i).
\end{proof}

We need the following critical consequence of the Baum--Katz theorem. It is
the point at which the geometry of the volume sequence enters.

\begin{lemma}
\label{lem:polynomial-baum-katz}
Let $q_d=1+1/d$, let $W_1,W_2,\ldots$ be i.i.d. with
$\E W_1=0$ and $\E|W_1|^{q_d}<\infty$, and let
$S_n=\sum_{j=1}^nW_j$ and $M_n=\max_{1\leq j\leq n}|S_j|$.
For every $A,u>0$,
\begin{equation}\label{eq:polynomial-baum-katz}
	\sum_{L=1}^{\infty}
	\Pp\bigl(M_{\lceil AL^d\rceil}>uL^d\bigr)<\infty.
\end{equation}
\end{lemma}

\begin{proof}
The maximal form of the Baum--Katz theorem
\cite{BaumKatz1965,ChenYiSung2015}, with its parameters $p=1$ and $r=q_d$,
states that for every $a>0$,
\begin{equation}\label{eq:baum-katz}
	\sum_{n=1}^{\infty}n^{q_d-2}\Pp(M_n>an)<\infty.
\end{equation}
Set $A_0=\max\{A,1\}$ and $N_L=\lceil A_0L^d\rceil$. There is a constant
$C<\infty$ such that $N_{L+1}\leq CL^d$ for all $L$. Hence, for every
$n\in\{N_L,\ldots,N_{L+1}-1\}$,
\[
	\{M_{N_L}>uL^d\}\subseteq\{M_n>(u/C)n\}.
\]
Moreover, because $q_d-2=-1+1/d$,
\begin{equation}\label{eq:critical-block-weight}
	\inf_{L\geq L_0}
	\sum_{n=N_L}^{N_{L+1}-1}n^{q_d-2}>0
\end{equation}
for some $L_0$. Indeed, when $d>1$ there are constants $c_1,c_2>0$ such
that $N_{L+1}-N_L\geq c_1L^{d-1}$ and $N_{L+1}\leq c_2L^d$ for all large
$L$. Since $q_d-2\leq0$, the sum in \eqref{eq:critical-block-weight} is
at least
\[
	c_1L^{d-1}(c_2L^d)^{q_d-2},
\]
and the power of $L$ is zero. When $d=1$ the weights equal one and the
interval contains at least one integer. Multiplying the preceding event
inclusion by the weights in
\eqref{eq:critical-block-weight}, summing over $L$, and using the disjoint
intervals in \eqref{eq:baum-katz} proves
\eqref{eq:polynomial-baum-katz}. Replacing $A$ by $A_0$ only enlarges the
maximum.
\end{proof}

\begin{theorem}[Sharp dimension-dependent moment criterion]
\label{thm:moment-criterion}
If
\begin{equation}\label{eq:moment-criterion}
	\E|J|^{1+1/d}<\infty,
\end{equation}
then the periodic pressure laws converge completely. Equivalently, the
periodic pressures converge almost surely to $p_{\infty,\nu}(\beta)$ under
every joint coupling of the finite-volume i.i.d. disorder arrays.
\end{theorem}

\begin{proof}
The assertion is immediate when $\beta=0$, so assume $\beta>0$, let
$q_d=1+1/d$, and write $Y_e=|J_e|$ and $m_1=\E Y_e$.

We first prove complete convergence for free boundary conditions. Fix a
block size $k$ and use the decomposition in \cref{sec:i.i.d.}. Enumerate the
$r^d$ blocks in shell order and set
\[
	X_z^{(k)}=\log Z_{B_z}^{\fr},\qquad
	\mu_k=\E X_z^{(k)}=k^dp_{k,\nu}^{\fr}.
\]
The variables $X_z^{(k)}$ are i.i.d.. By \eqref{eq:basic-bound},
$\E|X_z^{(k)}|^{q_d}<\infty$. Define the deterministic decoupled mean
\[
	a_{L,k}=\frac{r^d\mu_k+|R_{L,k}|\log2}{L^d}=\left(\frac{rk}{L}\right)^dp_{k,\nu}^{\fr} +\frac{|R_{L,k}|}{L^d}\log2.
\]
The coefficient comparison and 
$\sum_{e\in D_{L,k}}Y_e=m_1|D_{L,k}|+\sum_{e\in D_{L,k}}(Y_e-m_1)$ give
\begin{align}
	|\cP_L^{\fr}-p_{\infty,\nu}|
	\leq\frac{1}{L^d}
	\left|\sum_{z}(X_z^{(k)}-\mu_k)\right|
	+|a_{L,k}-p_{\infty,\nu}|+\beta m_1\frac{|D_{L,k}|}{L^d}
	+\frac{\beta}{L^d}
	\left|\sum_{e\in D_{L,k}}(Y_e-m_1)\right|.
	\label{eq:endpoint-free-decomposition}
\end{align}

Fix $t>0$. By \eqref{eq:free-rate}, choose $k$ so large that
$2\beta dm_1/k<t/4$. Since $a_{L,k}\to p_{k,\nu}^{\fr}$ and
\eqref{eq:deleted-count} holds, for all sufficiently large $L$ the two
deterministic terms on the right of
\eqref{eq:endpoint-free-decomposition} have sum less than $t/2$.
Consequently,
\begin{align}
	\Pp(|\cP_L^{\fr}-p_{\infty,\nu}|>t)\leq\Pp\left(
	\left|\sum_z(X_z^{(k)}-\mu_k)\right|>\frac{tL^d}{4}\right)+\Pp\left(
	\left|\sum_{e\in D_{L,k}}(Y_e-m_1)\right|
	>\frac{tL^d}{4\beta}\right).
	\label{eq:endpoint-free-tail}
\end{align}
The first sum contains $r^d\leq L^d$ centered i.i.d. variables with finite
$q_d$th moment. The deleted-edge count is at most
$(d/k+2d^2)L^d$ for $L\geq k$, and the variables $Y_e-m_1$ are centered i.i.d.
with finite $q_d$th moment. Each probability on the right of
\eqref{eq:endpoint-free-tail} is therefore bounded by an event of the form
in \cref{lem:polynomial-baum-katz}. This is a distributional statement:
after relabeling, a sum over any deterministic set of $N$ edges has the law
of the first $N$ partial sum, so no nesting of the sets $D_{L,k}$ across
volumes is assumed. Summing over $L$ proves
\begin{equation}\label{eq:free-complete}
	\sum_{L=1}^{\infty}
	\Pp(|\cP_L^{\fr}-p_{\infty,\nu}|>t)<\infty.
\end{equation}

It remains to restore the periodic seam. Couple the free and periodic
systems with the same disorder field and let $S_L$ be the $dL^{d-1}$ seam
edges. By \cref{lem:comparison},
\[
	\left|\cP_L^{\per}-\cP_L^{\fr}\right|\leq\frac{\beta dm_1}{L}+\frac{\beta}{L^d}\left|\sum_{e\in S_L}(Y_e-m_1)\right|.
\]
The deterministic term tends to zero. The random sum has the same law as
a partial sum of $dL^{d-1}\leq dL^d$ centered i.i.d. variables, so
\cref{lem:polynomial-baum-katz} yields
\[
	\sum_{L=3}^{\infty}\Pp(|\cP_L^{\per}-\cP_L^{\fr}|>t)<\infty.
\]
Combining this with \eqref{eq:free-complete} proves complete convergence of
the periodic pressures. The assertion for every joint coupling follows
from \cref{thm:universal-coupling}.
\end{proof}

\begin{corollary}\label{cor:complete}
Condition \textup{(ii)} of \cref{thm:universal-coupling} holds if either
\begin{enumerate}[label=\textup{(\alph*)}]
\item $|J|\leq K$ almost surely for some $K<\infty$, or
\item $J$ is Gaussian with finite variance $s^2$.
\end{enumerate}
\end{corollary}

\begin{proof}
The case $\beta=0$ is deterministic, so suppose $\beta>0$. For bounded couplings, changing one of the $dL^d$ coordinates changes the pressure by at most $2\beta K/L^d$. McDiarmid's inequality \cite{McDiarmid1989} gives
\[
	\Pp\bigl(|\cP_L^{\per}-p_{L,\nu}^{\per}|\geq t\bigr)\leq2\exp\!\left(-\frac{t^2L^d}{2d\beta^2K^2}\right).
\]
For Gaussian couplings, the Lipschitz constant with respect to the standard
Gaussian coordinates is
\[
	s\beta\sqrt d\,L^{-d/2}.
\]
Gaussian concentration \cite[Chapter~5]{BLM2013} gives the same bound with
$K$ replaced by $s$. These estimates are summable in $L$, and
\eqref{eq:i.i.d.-rates} transfers them from $p_{L,\nu}^{\per}$ to
$p_{\infty,\nu}$.
\end{proof}

Together with the next theorem, \cref{thm:moment-criterion} identifies the
sharp threshold on the moment scale.

\begin{theorem}\label{thm:heavy-tail}
Fix $d\geq1$ and $1\leq q<1+1/d$. There is a centered symmetric law
$\nu$ with $\E|J|^q<\infty$ such that, for every fixed $\beta>0$,
\begin{enumerate}[label=\textup{(\roman*)}]
\item the canonical nested free and periodic pressures converge almost surely to $p_{\infty,\nu}(\beta)$;
\item periodic pressures built from arrays independent across $L$ converge in probability but not almost surely.
\end{enumerate}
For $q=1$, the same law works simultaneously for all $d\geq1$.
\end{theorem}

\begin{proof}
Let $X\geq0$ have survival function
\begin{equation}\label{eq:heavy-tail-law}
	\Pp(X>t)=\frac{1}{(1+t)^q[\log(e+t)]^2},\qquad t\geq0,
\end{equation}
and let $J=\varepsilon X$, where $\varepsilon$ is an independent symmetric sign. The right side of \eqref{eq:heavy-tail-law} is continuous and nonincreasing, equals one at zero, and tends to zero. Moreover,
\[
	\E X^q=q\int_0^\infty t^{q-1}\Pp(X>t)\,dt<\infty,
\]
because the integrand is asymptotic to $q/[t(\log t)^2]$. Thus $J$ is centered and symmetric, and its first absolute moment $m_1$ is finite.

Fix $d$ and $\beta>0$, and set
\[
	A=d m_1+\frac{\log2+1}{\beta}.
\]
The finite-volume upper bound \eqref{eq:basic-bound} and \cref{thm:i.i.d.} give
$p_{\infty,\nu}(\beta)\leq\log2+\beta d m_1$. Let $M_L$ be the largest absolute coupling among the $dL^d$ periodic edges. Since the periodic graph is simple for $L\geq3$, \cref{lem:max-bond} gives $\cP_L^{\per}\geq\left(\beta M_L\right)/L^d$.
Consequently,
\begin{equation}\label{eq:spike-inclusion}
	\{M_L>AL^d\}
	\subseteq\{\cP_L^{\per}>p_{\infty,\nu}(\beta)+1\}.
\end{equation}

Let $n_L=dL^d$ and $y_L=\Pp(X>AL^d)$. The explicit tail gives
$n_Ly_L=O(L^{-d(q-1)}(\log L)^{-2})$, hence $n_Ly_L\leq 1$ for a large $L$. For such a $L$,
\[	
	\Pp(M_L>AL^d)=1-(1-y_L)^{n_L}\geq 1-e^{-n_Ly_L}
	\geq\frac{n_Ly_L}{2}.
\]
Also, for suitable constants $C,c>0$ and all large $L$,
\[
	1+AL^d\leq2AL^d,\qquad
	\log(e+AL^d)\leq C\log L,\qquad
	n_Ly_L\geq\frac{cL^{-d(q-1)}}{(\log L)^2}.
\]
Since $d(q-1)<1$, choose $\eta>0$ with $d(q-1)+\eta<1$.
The inequality $(\log L)^2\leq L^\eta$ holds eventually, so the last lower
bound has a divergent sum. By \eqref{eq:spike-inclusion}, the pressure
deviations are not summable. If the arrays are independent across $L$, the
second Borel--Cantelli lemma produces infinitely many such deviations,
proving failure of almost sure convergence. The one-volume marginals
nevertheless satisfy \eqref{eq:marginal-convergence}, so convergence in
probability holds. Canonical almost sure convergence follows from
\cref{thm:i.i.d.}. When $q=1$, \eqref{eq:heavy-tail-law} is independent of $d$,
which proves the last assertion.
\end{proof}

\section{Sharpness of the first moment in one dimension}\label{sec:sharpness}

The first absolute moment used in the baseline and transport theorems is a genuine threshold in dimension one.

\begin{proposition}\label{prop:one-dimensional}
For the one-dimensional free chain and $\beta>0$,
\begin{equation}\label{eq:one-dimensional-pressure}
	\cP_L^{\fr}(\beta,J)
	=\log2+\frac{1}{L}\sum_{i=0}^{L-2}\log\cosh(\beta J_i).
\end{equation}
If $\E|J|<\infty$, the almost-sure limit is
$\log2+\E\log\cosh(\beta J)$. If $\E|J|=\infty$, then the quenched finite-volume pressure is infinite in expectation for $L\geq2$, and
$\cP_L^{\fr}(\beta,J)\to\infty$ almost surely.
\end{proposition}

\begin{proof}
Successively summing over endpoint spins gives
\[
	Z_L^{\fr}=2^L\prod_{i=0}^{L-2}\cosh(\beta J_i),
\]
which proves \eqref{eq:one-dimensional-pressure}. Since
\[
	\beta|x|-\log2\leq\log\cosh(\beta x)\leq\beta|x|,
\]
the expectation of $\log\cosh(\beta J)$ is finite exactly when $\E|J|$ is finite. The ordinary strong law proves the finite-limit assertion. In the infinite-mean case, apply the strong law to
$\min\{\log\cosh(\beta J_i), M\}$ and then send $M\to\infty$ by monotone convergence.
\end{proof}

\section{Discussion}

The principal conceptual point is that the joint construction across volumes
is part of an almost-sure thermodynamic-limit statement. Two constructions
may have the same disorder law at every fixed $L$, and hence identical
one-volume pressure laws, quenched means, and convergence in probability,
while producing different recurrence of exceptional fluctuations. A
pathwise theorem must therefore specify whether the arrays are nested,
independently resampled, or quantified over all joint couplings. The
universal-coupling criterion isolates precisely this dependence.

The exponent $1+1/d$ comes from converting Baum--Katz summability in the
sample size $n$ into summability over side lengths $L$. If
$N_L\asymp L^d$, then $N_{L+1}-N_L\asymp L^{d-1}$, while the Baum--Katz
weight on this block satisfies $n^{q-2}\asymp L^{d(q-2)}$. The total block
weight is therefore of order $L^{d-1+d(q-2)}$, which is of constant order
exactly when $q=1+1/d$. This block estimate is the endpoint step and cannot
be obtained merely by taking a limit of noncritical bounds.

The condition $\E|J|^{1+1/d}<\infty$ is a distribution-free sufficient
condition for complete convergence and hence for almost-sure convergence
under every joint coupling. The counterexamples show that no smaller power
moment gives such a uniform guarantee over centered symmetric laws. They do
not show that every fixed law with an infinite critical moment fails complete
convergence. A law-by-law criterion would require finer tail information,
including slowly varying corrections, and is not addressed here.

The deterministic comparison, transport estimates, and complete-convergence
upper-bound arguments extend to finite-state nearest-neighbor models with
uniformly bounded local observables. By contrast, the large-bond lower bound
used for sharpness is proved here only for the Ising interaction. Our
sharpness claim is therefore confined to the Edwards--Anderson Ising model.
Extensions to genuine multi-body interactions are not asserted.

\section*{Declarations}

\subsection*{Funding}
The authors declare that no funds, grants, or other support were received during the preparation of this manuscript.

\subsection*{Competing interests}
The authors have no relevant financial or non-financial interests to disclose.

\subsection*{Authors' contributions}
All authors contributed to the conceptualization, mathematical analysis, proof verification, and writing of the manuscript. All authors contributed equally to this work. All authors read and approved the final manuscript.

\subsection*{Data availability}
Data sharing is not applicable because no datasets were generated or analyzed in this theoretical study.

\subsection*{Ethical Approval/Human and Animal Rights}
This article does not contain any studies with human participants or animals performed by any of the authors.

\Addresses

\end{document}